\documentclass[10pt]{article}
\usepackage{color, amssymb, amsthm, amsmath, amsfonts, ascmac, comment, enumerate}

\usepackage{hyperref}
\usepackage{graphicx}
\usepackage{multirow}
\usepackage{tcolorbox}
\usepackage{xcolor}

\newtheorem{theorem}{Theorem}
\newtheorem{proposition}{Proposition}	

\newtheorem{corollary}{Corollary}

\newtheorem{definition}{Definition}

\newcommand{\Bset}{\{0, 1\}}
\newcommand{\QCC}{\mathrm{QCC}}

\newcommand{\Co}{\mathrm{Co}}
\usepackage[top=30truemm,bottom=30truemm,left=20truemm,right=20truemm]{geometry}

\title{Bounds on oblivious multiparty quantum communication complexity}
\author{Fran\c{c}ois~Le Gall \and Daiki Suruga}

\begin{document}

\maketitle

\begin{abstract}
The main conceptual contribution of this paper is investigating quantum multiparty communication complexity in the setting where communication is \emph{oblivious}. This requirement, which to our knowledge is satisfied by all quantum multiparty protocols in the literature, means that the communication pattern, and in particular the amount of communication exchanged between each pair of players at each round is fixed \emph{independently of the input} before the execution of the protocol. 
We show, for a wide class of functions, how to prove strong lower bounds on their oblivious quantum $k$-party communication complexity using lower bounds on their \emph{two-party} communication complexity. 
We apply this technique to prove tight lower bounds for all symmetric functions with \textsf{AND} gadget, and in particular obtain an optimal $\Omega(k\sqrt{n})$ lower bound on the oblivious quantum $k$-party communication complexity of the $n$-bit Set-Disjointness function. We also show the tightness of these lower bounds by giving (nearly) matching upper bounds.
\end{abstract}

\section{Introduction}\label{sec:intro}
\subsection{Background}
\paragraph{Communication complexity.}
Communication complexity, first introduced by Yao in a seminal paper~\cite{Yao79} to investigate circuit complexity, has become a central concept in theoretical computer science with a wide range of applications (see \cite{KN96,RY20} for examples). In its most basic version, called two-party (classical) communication complexity, two players, usually called Alice and Bob, exchange (classical) messages in order to compute a function of their inputs. More precisely, Alice and Bob are given inputs $x_1 \in \{0, 1\}^n$ and $x_2\in \{0, 1\}^n$, respectively, and their goal is to compute a function $f: (x_1, x_2) \mapsto \{0, 1\}$ by communicating with each other, with as little communication as possible.
\par
There are two important ways of generalizing the classical two-party communication complexity: one is to consider classical \textit{multiparty} communication complexity and the other one is to consider \textit{quantum} two-party communication complexity.
In (classical) multiparty communication complexity, there are~$k$ players $P_1$, $P_2$, $\ldots$, $P_k$, each player $P_i$ is given an input $x_i \in \{0, 1\}^n$. The players seek to compute a given function $f:(x_1,\ldots,x_k) \mapsto \{0, 1\}$ using as few (classical) communication as possible.\footnote{This way of distributing inputs is called the number-in-hand model. There exists another model, called the number-on-the-forehead model, which we do not consider in this paper.} 
The other way of generalizing the classical two-party communication complexity is \textit{quantum} two-party communication complexity, where Alice and Bob are allowed to use \textit{quantum} communication, i.e., they can exchange messages consisting of quantum bits.
Since its introduction by Yao \cite{Yao93}, the notion of quantum two-party communication complexity has been the subject of intensive research in the past thirty years, which lead to several significant achievements, e.g., \cite{BCWdW01,BCW98,BC97,CvDNT98,Tou15,Yao93}.
\par
In this paper, we consider both generalizations simultaneously: we consider quantum multiparty communication complexity for $k>2$ parties.  This generalization has been the subject of several works \cite{BvDHT99,LN18,LSS09,TNY09} but, compared to the two-party case, is still poorly understood.

\paragraph{Set-Disjointness.}
One of the most studied functions in communication complexity is $\operatorname{Set-Disjointness}$. For any $k\ge 2$ and any $n\ge 1$, the $k$-party $n$-bit Set-Disjointness function, written $\operatorname{DISJ}_{n, k}$, has for input a $k$-tuple $(x_1,\ldots,x_k)$, where $x_i\in\{0,1\}^n$ for each $i\in\{1,\ldots,k\}$. The output is $1$ if there exists an index $j\in\{1,\ldots,n\}$ such that $x_1[j]=x_2[j]=\cdots=x_k[j]=1$, where $x_i[j]$ denotes the $j$-th bit of the string $x_i$, and $0$ otherwise. The output can thus be written as
\[
\operatorname{DISJ}_{n, k}(x_1,\ldots,x_k)=\bigvee_{j=1}^{n} (x_1[j] \wedge \cdots \wedge x_k[j]).
\]

$\operatorname{Set-Disjointness}$ plays a central role in communication complexity since a multitude of problems can be analyzed via a reduction from or to this function (see \cite{CP10} for a good survey). In the two party classical setting, the communication complexity of $\operatorname{Set-Disjointness}$ is $\Theta(n)$: while the upper bound $O(n)$ is trivial, the proof of the lower bound $\Omega(n)$, which holds even in the randomized setting, is highly non-trivial \cite{KS92,Raz92}. The $k$-party $\operatorname{Set-Disjointness}$ function with $k>2$ has received much attention as well, especially since it has deep applications to distributed computing \cite{DKO12}. Proving strong lower bounds on multiparty communication complexity, however, is significantly more challenging than in the two-party model. After much effort, a tight lower bound for $k$-party $\operatorname{Set-Disjointness}$ was nevertheless obtained in the classical setting: recent works~\cite{BEO+13,RU19}  were able to show a lower bound $\Omega(kn)$ for $\operatorname{DISJ}_{n, k}$, which is (trivially) tight. 

In the quantum setting, Buhrman et al.~\cite{BCW98} showed that the two-party quantum communication complexity of the $\operatorname{Set-Disjointness}$ function is $O(\sqrt{n}\log n)$, which gives a nearly quadratic improvement over the classical case. The logarithmic factor was then removed by Aaronson and Ambainis \cite{AA03}, who thus obtained an $O(\sqrt{n})$ upper bound. A matching lower bound $\Omega(\sqrt{n})$ was then proved by Razborov \cite{Raz03}. For $k$-party quantum communication complexity, an $ O(k\sqrt{n}\log n)$ upper bound is easy to obtain from the two-party upper bound from~\cite{BCW98}.\footnote{We will show later (in Theorem~\ref{Thm_DISJ_upper} in Section \ref{sec_match_bound}) how to obtain an improved $O(k\sqrt{n})$ upper bound based on the protocol from \cite{AA03}.} An important open problem, which is fundamental to understand the power of quantum distributed computing, is showing the tightness of this upper bound. In view of the difficulty in proving the $\Omega(kn)$ lower bound in the classical setting, proving a $\Omega(k\sqrt{n})$ lower bound in the quantum setting is expected to be challenging.

\subsection{Our contributions}\label{subsec_models}
\paragraph{Our model.}
The main conceptual contribution of this paper is investigating quantum multiparty communication complexity in the setting where communication is \emph{oblivious}. 
This requirement means that the communication pattern, and in particular the amount of communication exchanged between each pair of players at each round is fixed \emph{independently of the input} before the execution of the protocol. 
(See Section~\ref{subsec_model} for the formal definition.) This requirement is widely used in classical networking systems~(e.g., \cite{FW98,MN93,RS19}) and classical distributed algorithms (e.g., \cite{Censor-Hillel+19}), and to our knowledge is satisfied by all known quantum communication protocols (for any problem) that have been designed so far. It has also been considered in the quantum setting by Jain et al.~\cite[Result~3]{JRS03}, who gave an $\Omega(n / r^2)$ bound on the quantum communication complexity of $r$-round $k$-party oblivious protocols for a promise version of $\operatorname{Set-Disjointness}$.

\paragraph{Our results.}\label{subsec_results}
The main result of this paper holds for a class of functions which has a property that we call \textit{$k$-party-embedding}.
We say that a $k$-player function $f_{k}$ is a $k$-party-embedding function of a two-party function $f_2$ if the function $f_{2}$ can be ``embedded'' in $f_{k}$ by embedding the inputs of $f_2$ in \emph{any} position among the inputs of $f_k$. Many important functions such as any $k$-party symmetric function (including as important special cases the Set-Disjointness function $\operatorname{DISJ}_{n, k}$ and the $k$-party Inner-Product function) or the $k$-party equality function have this property. For a formal definition of the embedding property, we refer to Definition~\ref{Def_reduction} in Section \ref{sec:lb}. Our main result is as follows.
\vspace{3mm}
\newline
\textbf{Theorem~\ref{Thm_reduction} (informal)}
\textit{
Let $f_{k}$ be a $k$-party function that is a $k$-party-embedding function of a two-party function $f_{2}$. Then the oblivious $k$-party quantum communication complexity of $f_k$ is at least $k$ times the two-party quantum communication complexity of $f_{2}$. 
}
\vspace{3mm}

Theorem~\ref{Thm_reduction} enables us to prove strong lower bounds on oblivious quantum $k$-party communication complexity using the quantum two-party communication complexity.
\footnote{Note that in the two-party setting, the notions of oblivious communication complexity and non-oblivious communication complexity essentially coincide, since any non-oblivious communication protocol can be converted into an oblivious communication protocol by increasing the complexity by a factor at most  two.
To see this, without loss of generality assume that each player sends only one qubit at each round.}
This is useful since two-party quantum communication complexity is a much more investigated topic than $k$-party quantum communication complexity, and many tight bounds are known in the two-party setting. 
For example, we show how to use Theorem~\ref{Thm_reduction} to analyze the oblivious quantum $k$-party  communication complexity of $\operatorname{DISJ}_{n, k}$ and obtain a tight $\Omega(k\sqrt{n})$ bound:
\begin{corollary}\label{cor:DISJlb}
In the oblivious communication model, the $k$-party quantum communication complexity of $\operatorname{DISJ}_{n, k}$ is $\Omega(k\sqrt{n})$.
\end{corollary}
 More generally, Theorem~\ref{Thm_reduction} enables us to derive tight bounds for the  oblivious quantum $k$-party communication complexity of arbitrary symmetric functions. Since symmetric functions play an important role in communication complexity \cite{CCH+22,Pat92,Raz03,She11,SZ09}, our results might thus have broad applications. Additionally, we also give lower bounds for non-symmetric functions that have the $k$-party-embedding property, such as the equality function. Our results are summarized in Table~\ref{table_our_results}.

To complement our lower bounds, we show tight (up to possible poly-log factors) upper bounds for these functions. The upper bounds are summarized in Table~\ref{table_our_results} as well. Note that if we apply our generic $O(k\log n \cdot G_n(f))$ bound in Table~\ref{table_our_results} to $\mathrm{DISJ}_{n, k}$, we only get the upper bound $O(k\log n\cdot \sqrt{n})$. We thus prove directly an optimal $O(k \sqrt{n})$ upper bound (Theorem~\ref{Thm_DISJ_upper}) by showing how to adapt the optimal two-party protocol from \cite{AA03} to the $k$-party setting.

\begin{table}[hbtp]
  \centering
  \begin{tabular}{|c|ccc|ccc|}
    \hline
    Functions & \multicolumn{3}{c|}{$2$-party protocols}  & \multicolumn{3}{c|}{$k$-party oblivious protocols}   \\
    \hline
    &Lower &~& Upper &Lower &~& Upper\\
    \hline \hline
    \multirow{2}{*}{Symmetric functions} &$\Omega( G_n(f) )$&& $O(\log n \cdot G_n(f))$
	&$\Omega(k \cdot G_n(f))$ && $O(k \log n \cdot G_n(f))$   \\
	& in \cite{Raz03} && in \cite{Raz03}& Proposition~\ref{Pro_sym_lower}&& Theorem~\ref{Thm_sym_upper}\\
    \hline
    \multirow{2}{*}{Set-Disjointness}&$\Omega(\sqrt{n})$ && $O(\sqrt{n})$ 
	&$\Omega(k \sqrt{n})$ && $O(k  \sqrt{n})$\\
	&in \cite{Raz03}&&in \cite{AA03}&Corollary~\ref{cor:DISJlb}&& Theorem~\ref{Thm_DISJ_upper}\\
    \hline
    Set-Disjointness&\multirow{2}{*}{$\tilde{\Omega} (n/M)$}&& \multirow{2}{*}{$O(n/M)$} &\multirow{2}{*}{$\tilde{\Omega} (k \cdot n / M)$ }&& 
	\multirow{2}{*}{$O(k \cdot n/M)$} \\
    in $M$-round &&&&&&\\
    ($M \leq O(\sqrt{n})$) &in \cite{BGK+18}&&(folklore)&Proposition~\ref{Pro_DISJ_m_lower}&& Corollary~\ref{Cor_DISJ_m_upper}\\
    \hline
	\multirow{2}{*}{Equality function}&$\Omega(1)$&&$O(1)$&$\Omega(k)$&&$O(k)$\\
	&(trivial)&&e.g., \cite{KN96}&Proposition~\ref{Pro_Eq_lower}&&Proposition~\ref{Prop_Eq_upper}\\
    \hline
  \end{tabular}\vspace{2mm}
  \caption{Our results for oblivious quantum $k$-party communication complexity, along with known bounds for the two-party setting. For a symmetric function $f$, the notation  $G_n(f)$ refers to the quantity defined in Equation (\ref{eq:sym}).
}
  \label{table_our_results}
\end{table}

\section{Models of Quantum Communication}\label{sec_models}
\textbf{Notations:}
All logarithms are base $2$ in this paper. We denote $[k] = \{1, \ldots, k\}$.
For any set $\mathcal{X}$ and $k \geq 1$, $\mathcal{X}^k := \underbrace{\mathcal{X} \times \cdots \times \mathcal{X}}_{k}$.
\par
Here we formally define the quantum multiparty communication model.
As mentioned in Section~\ref{subsec_models}, this communication model
satisfies the oblivious routing condition (or simply the oblivious condition), meaning that
the number of qubits used in communication at each round
is predetermined (independent of inputs, private randomness, public randomness and outcome of quantum measurements).
Since details of the model are important especially when proving lower bounds, we explain the model in detail below.

\subsection{Quantum multiparty communication model}\label{subsec_model}
In $k$-party quantum communication model, at each round, players are allowed to send quantum messages\footnote{Trivially, players can send classical messages using quantum communication in this communication model.} to all of the players
but
the number of qubits used in communication is predetermined.
This condition is called \emph{oblivious}.
Therefore for any $k$-player $M$-round protocol $\Pi$, we define the functions
$C_{P_i \to P_j}:[M] \to \mathbb{N}\cup\{0\}~(i, j \in [k])$ which represent the number of qubits $C_{P_i \to P_j}(m)$
transmitted at $m$-th round from $i$-th player to $j$-th player.
\par
\textbf{Procedure:}
Before the execution of the protocol, all players $P_1, \ldots, P_k$ share an entangled state or public randomness if needed.
Each player $P_i$ is then given an input.
At each round $m \leq M$, every player $P_i$ performs some operations (such as unitary operations, measurements, coin flipping) onto $P_i$'s register
and send $C_{P_i \to P_1}(m)$ qubits to the player $P_1$, $C_{P_i \to P_2}(m)$ qubits to the player $P_2$, $\cdots$ , 
and $C_{P_i \to P_k}(m)$ qubits to the player $P_k$.
All messages from all players are sent simultaneously.
This continues until $M$-th round is finished.
Finally, each player $P_i$ output the answer based on the contents of $P_i$'s register.
\par
We define the communication cost of this protocol as
\begin{equation*}
\mathrm{QCC}(\Pi) := \sum_{m \in [M]} \sum_{\substack{i, j \in [k]\\ i \neq j}} C_{P_i \to P_j}(m).
\end{equation*}

\subsection{Coordinator model}
Let us also describe the definition of the following coordinator model 
so that discussions on the upper bounds in Section~\ref{sec_match_bound} become simpler.
\par
In $k$-party coordinator model, there are $k$-players, each is given an input, and another player called a coordinator who is not given any input.
Each player can communicate only with the coordinator.
Similar to the ordinary communication model, the number of qubits used in communication is predetermined.
Therefore for any $k$-player $M$-round protocol $\Pi$ in coordinator model, we define the functions
$C_{P_i \to \mathrm{Co}}:[M] \to \mathbb{N}\cup\{0\}$ and  $C_{\mathrm{Co}\to P_i}:[M] \to \mathbb{N}\cup\{0\} $ for $i \leq k$.
The value $C_{P_i \to \mathrm{Co}}(m)$ (resp. $C_{\mathrm{Co}\to P_i}(m)$) represent the number of qubits 
transmitted at $m$-th round from $i$-th player to the coordinator (resp. the coordinator to $i$-th player).
\par
\textbf{Procedure:} 
Before the execution of the protocol, all players $P_1, \ldots, P_k$ and the coordinator  share an entangled state or public randomness if needed.
Each player $P_i$ is then given input.
At each round $m \leq M$, each player $P_i$ performs some operations onto $P_i$'s register
and send $C_{P_i \to \mathrm{Co}}(m)$ qubits to the coordinator.
After that, the coordinator, who received $C_{P_1 \to \mathrm{Co}}(m) + \cdots + C_{P_k \to \mathrm{Co}}(m)$ qubits, 
performs some operations (such as unitary operations, measurements, coin flipping) onto the coordinator's register 
and sends back $C_{\mathrm{Co} \to P_i}(m)$ qubits to each player $P_i$.
This continues until the $M$-th round is finished.
Finally, each player $P_i$ outputs the answer based on the contents of $P_i$'s register.
\par
We define the communication cost of this protocol as
$\mathrm{QCC}_\mathrm{Co}(\Pi) := \sum_{m \in [M]}$ $\sum_{i \in [k]}  C_{P_i \to \mathrm{Co}}(m) +  C_{\mathrm{Co} \to P_i}(m).$

\subsection{Protocol for computing a function}\label{subsec_protocol_computing}
We define a protocol computing a function as follows.
\begin{definition}
We say a protocol $\Pi$ computes $f: \mathcal{X}_1 \times \cdots \times \mathcal{X}_k \to \mathcal{Y}$ with error $\varepsilon \in [0, 1/2)$ if
\begin{equation*}
\forall i \in [k], ~
\forall x = (x_1, \ldots, x_k) \in \mathcal{X}_1 \times \cdots \times \mathcal{X}_k,  
\quad
\Pr(\Pi^i_\mathrm{out}(x) \neq f(x))\leq \varepsilon
\end{equation*}
where $\Pi^i_\mathrm{out}(x)$ denotes $P_i$'s output of the protocol 
on input $x$.
\end{definition}

We denote by $\mathcal{P}_k(f, \varepsilon)$  the set of $k$-party protocols computing a function $f$ with error $\varepsilon$ 
in the quantum multiparty communication model.
The quantum communication complexity of function $f$ with error $\varepsilon$ in the model is defined as
$\mathrm{QCC}(f, \varepsilon) $$:= \min_{\Pi \in \mathcal{P}_k(f, \varepsilon)} \mathrm{QCC}(\Pi)$.
\par
We also define the bounded round communication complexity of function $f$ as
$\mathrm{QCC}^M(f, \varepsilon) := \min_{\Pi \in \mathcal{P}^M_k(f, \varepsilon)} \mathrm{QCC}(\Pi)$
where we use the superscript $M$ to denote 
the set of $M$-round protocols $\mathcal{P}^M_k(f, \varepsilon)$.
Regarding the coordinator model, we define $\mathcal{P}_k(f, \varepsilon)_\Co, \mathrm{QCC}_\Co(f, \varepsilon), \mathcal{P}_k^M(f, \varepsilon)_\Co$, 
and $\mathrm{QCC}^M_\Co(f, \varepsilon)$ in similar manners as above. 
\par
As is easily seen\footnote{To show $\operatorname{QCC}^{2M}(f, \varepsilon) \leq \operatorname{QCC}^M_\Co(f, \varepsilon)$, assign $P_1$ the role of the coordinator. 
To show $\operatorname{QCC}^M_\Co(f, \varepsilon) \leq 2\operatorname{QCC}^M(f, \varepsilon)$, consider the coordinator only passes messages without performing any operation.},
$\operatorname{QCC}^{2M}(f, \varepsilon) \leq \operatorname{QCC}^M_\Co(f, \varepsilon) \leq 2\operatorname{QCC}^M(f, \varepsilon)$ holds.
This means the two models asymptotically have the same power even in bounded round setting.

\subsection{Symmetric functions}
A function $f: \{0, 1\}^n \times \{0, 1\}^n \to \{0, 1\}$ 
is symmetric\footnote{Although a function $f: \{0, 1\}^n \to \Bset$ is generally said to be symmetric when any permutation on the input does not change the value of $f$, 
in this paper we focus on functions of the form $f : \{0, 1\}^n \times \{0, 1\}^n \to \{0, 1\}$, and use the same definition for symmetric functions (predicates) as in \cite{Raz03}.} if there exists a function $D_f: [n]\cup\{0\} \to \{0, 1\}$ such that $f(x, y) = D_f(|x \cap y|)$,
where $x\cap y$ is the intersection of the two sets $x,y\subseteq[n]$ corresponding to the strings $x,y$. This means that the function $f$ depends only on the Hamming weight of (the intersection of) the inputs. For any symmetric function $f: \{0, 1\}^n \times \{0, 1\}^n \to \{0, 1\}$, let us write 
\begin{equation}\label{eq:sym}
G_n(f)= \sqrt{nl_0(D_f)} + l_1(D_f),
\end{equation}
where
\begin{eqnarray*}
l_0(D_f) &=& \max\big\{l \:|\:  1 \leq l\leq n/2 \text{~and~}D_f(l) \neq D_f(l - 1)\big\},\\
l_1(D_f) &=& \max\big\{n - l \: |\: n/2 \leq l < n \text{~and~}D_f(l) \neq D_f(l + 1)\big\}.
\end{eqnarray*}
Razborov~\cite{Raz03} showed the lower bound $\Omega(G_n(f))$ on the quantum two-party communication complexity of any symmetric function $f$, 
and also obtained a nearly matching upper bound $O(G_n(f)\log n)$. 
We also note that for any function $D_f$,  this function is constant on the interval $[l_0(D_f), n - l_1(D_f)]$ by the definitions of $l_0(D_f)$ and $l_1(D_f)$.
In Section~\ref{sec_upper_sym}, we use this fact to prove a nearly matching upper bound on the oblivious quantum multiparty communication model.

Analogously, 
a $k$-party function $f: \{0, 1\}^{n\cdot k}  \to \{0, 1\}$ is symmetric 
when represented as $f(x_1, \ldots, x_k) = D_f(|x_1 \cap \cdots \cap x_k|)$
using some function $D_f: [n]\cup\{0\} \to \{0, 1\}$.
The $k$-party $n$-bit Set-Disjointness function $\operatorname{DISJ}_{n, k}$ defined in Section \ref{sec:intro} is a symmetric function. 
The $k$-party $n$-bit (generalized) Inner-Product function $\mathrm{IP}_{n, k}$, defined for any $x_1,\ldots,x_k\in\{0,1\}^n$ as 
\[
\mathrm{IP}_{n, k}(x_1,\ldots,x_k) = 
(x_1[1]\wedge \cdots \wedge x_k[1])\oplus \cdots \oplus (x_1[n]\wedge \cdots \wedge x_k[n])
\]
is also symmetric.
\par
On the other hand, the $k$-party $n$-bit equality function $\mathrm{Equality}_{n, k}$, defined for any $x_1,\ldots,x_k\in\{0,1\}^n$ as 
\[
\mathrm{Equality}_{n, k}(x_1,\ldots,x_k) = 
\begin{cases}
1&\textrm{ if } x_1=x_2=\cdots=x_k,\\
0& \textrm{ otherwise},
\end{cases}
\]
is not symmetric.

\section{Lower bounds}\label{sec:lb}
Here we show Proposition~\ref{Prop_reduction}, which relates the oblivious communication complexity of a $k$-party function 
$f_k: \mathcal{X}^k  \to \mathcal{Y}$
to the oblivious communication complexity of a two-party function $\tilde{f}_2: \mathcal{X} \times \mathcal{X} \to \mathcal{Y}$
when $f_k$ is a \textit{$k$-party-embedding} function of $\tilde{f}_2$ in the following sense.

\begin{definition}\label{Def_reduction}
A function $f_k: \mathcal{X}^k \to \mathcal{Y}$
is a $k$-party-embedding function of $\tilde{f}_2: \mathcal{X} \times \mathcal{X} \to \mathcal{Y}$
 if for any $i \in [k]$,
there is a map $x_{-i}: \mathcal{X} \to \mathcal{X}^{k-1}$ such that
\begin{equation*}
\forall x_1, x_2 \in \mathcal{X} \quad \tilde{f}_2(x_1, x_2) = f_k([x_{-i}(x_2), i, x_1])
\end{equation*}
holds,
where $[y, i, x] := (y_1 \ldots, y_{i-1}, x, y_i, \ldots, y_{k-1})$ for $y = (y_i)_{i \leq k-1} \in \mathcal{X}^{k-1}$ and $x \in \mathcal{X}$.
\end{definition}
For example, $\mathrm{DISJ}_{n, k}$ ($k \geq 2$) is a $k$-party-embedding function of $\mathrm{DISJ}_{n, 2}$ 
because we can take $\mathcal{X} = \{0, 1\}^n$, $\mathcal{Y} = \{0, 1\}$
and $x_{-i}(x) = (x, 1^n, \ldots, 1^n)$. 

\par
Using this definition, we show the following proposition.
\begin{proposition}\label{Prop_reduction}
Let $f_k: \mathcal{X}^k \to \mathcal{Y}$ be a function and 
suppose $f_k$ is a $k$-party-embedding function of $\tilde{f}_2: \mathcal{X} \times \mathcal{X} \to \mathcal{Y}$ .
For any protocol $\Pi_k \in \mathcal{P}_k(f_k, \varepsilon)$, there is a two-party protocol $\tilde{\Pi} \in \mathcal{P}_2(\tilde{f}_2, \varepsilon)$ such that
$\mathrm{QCC}(\tilde{\Pi}) \leq \frac{2 \mathrm{QCC}(\Pi_k)}{k}$
holds.
\end{proposition}
\begin{proof}
\begin{figure}[htp]
\begin{tabular}{cc}
\begin{minipage}[t]{0.45\hsize}
\centering
\includegraphics[keepaspectratio, scale=0.18]{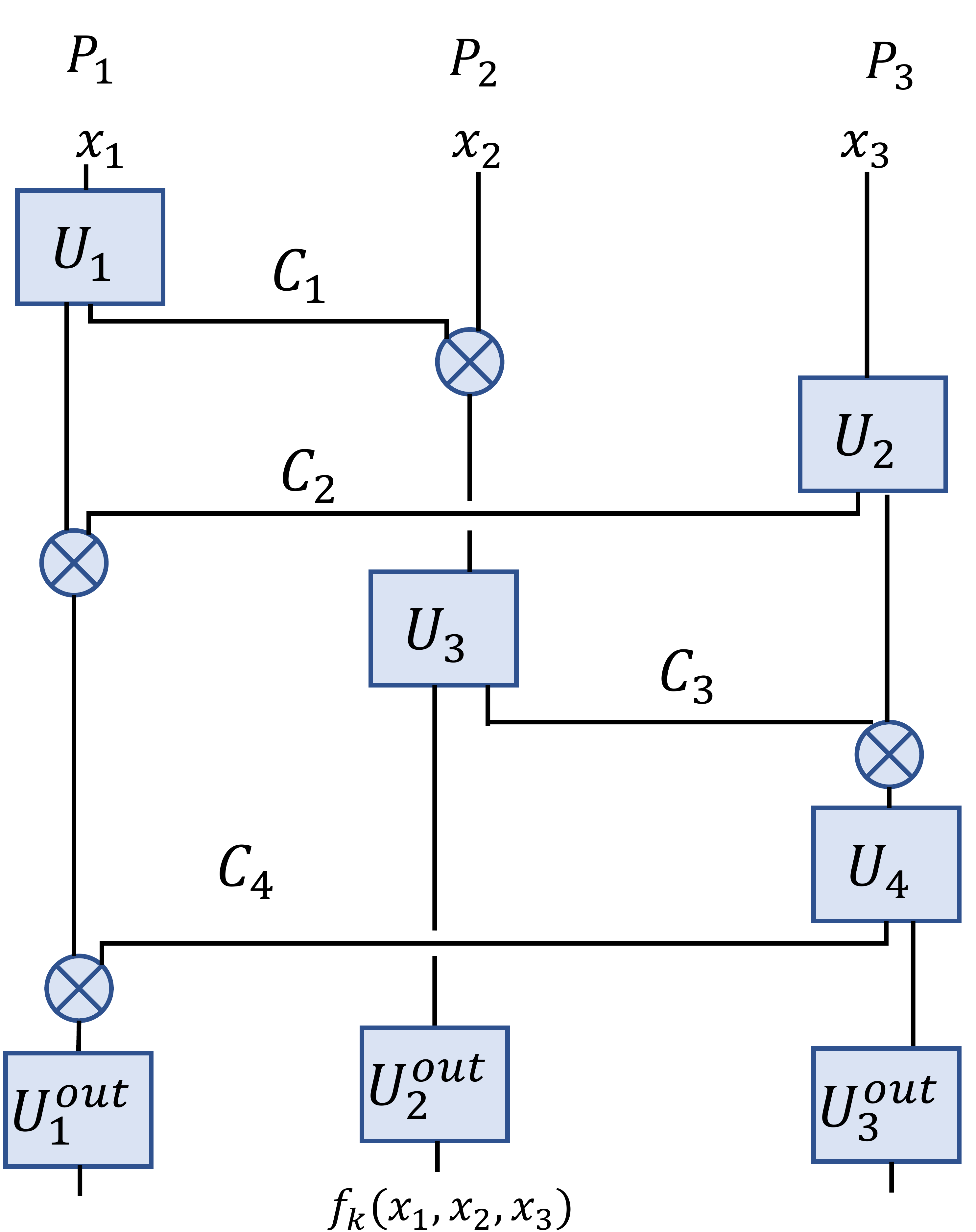}
\caption{Example of $\Pi_k$ for $f_k$ when $k = 3$. 
(Prior entanglement is omitted.)
Assume $\mathrm{QCC}_1(\Pi_k) \leq \mathrm{QCC}(\Pi_k) / k$, i.e., $i_0 = 1$.
}
\label{Fig_Pi_k}
\end{minipage} 
&\quad 
\begin{minipage}[t]{0.45\hsize}
\centering
\includegraphics[keepaspectratio, scale=0.18]{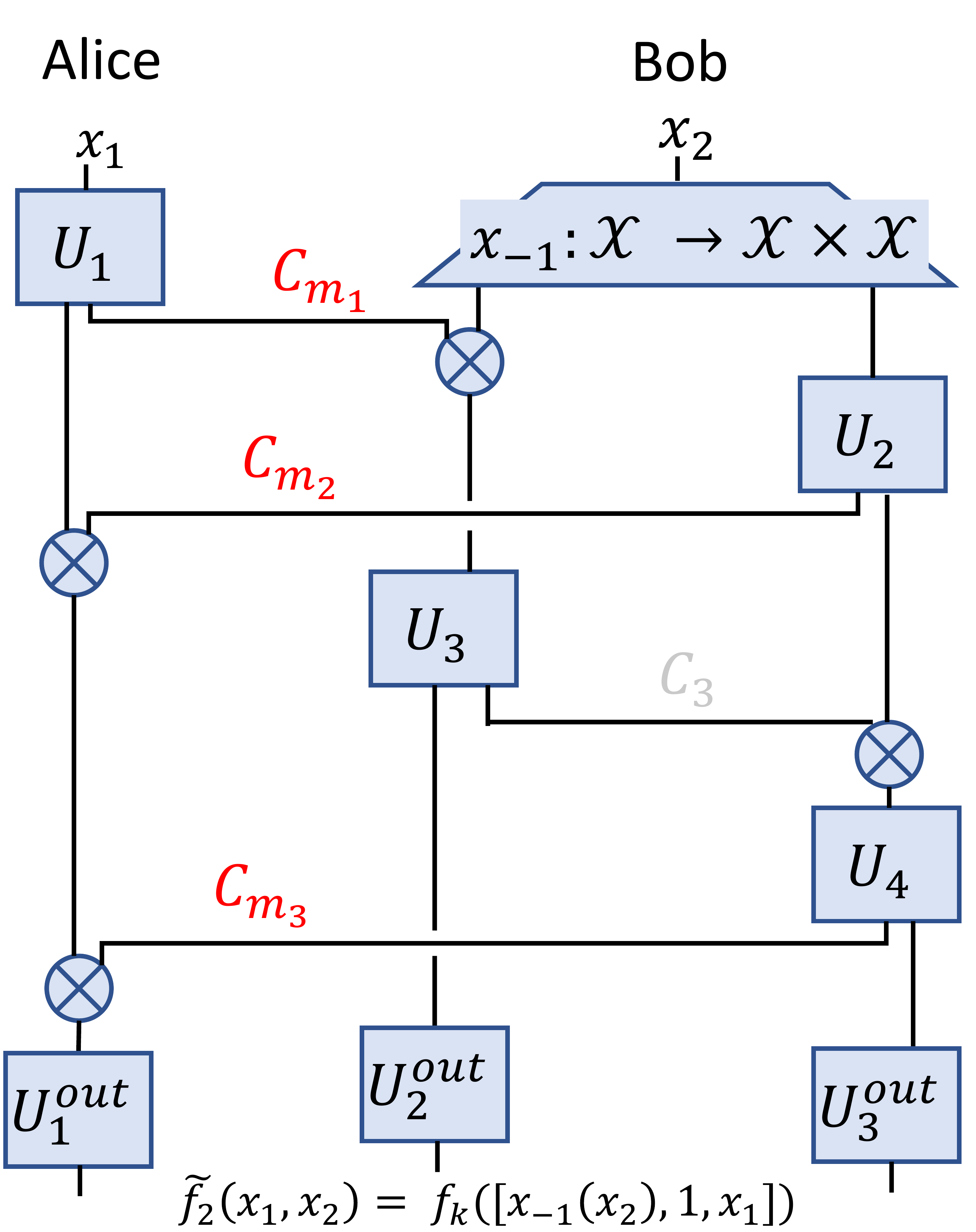}
\caption{Protocol $\tilde{\Pi}$ for $\tilde{f}_2$ created from $\Pi_k$ when $i_0 = 1$. 
Here, the communication $C_3$ is internally computed by Bob and the entire communication cost is $\mathrm{QCC}(\tilde{\Pi}) = \mathrm{QCC}_1(\Pi_k)$.
}
\label{Fig_Pi_2}
\end{minipage}
\end{tabular}
\end{figure}

Without loss of generality, we assume that at each round only one player sends a message in the protocol $\Pi_k$.
Let $\mathrm{QCC}_i(\Pi_k)$ denote the communication cost of player $i$, which we define as the sum of the number of qubits
exchanged, either sent or received, by player $i$.
For example, in Fig~\ref{Fig_Pi_k} showing an example\footnote{
In Fig~\ref{Fig_Pi_k}, $U_i$ and $U_i^\mathrm{out}$ denote classical or quantum operations and $\otimes$ denotes the operation of attaching registers.
$U_i^\mathrm{out}$ usually includes measurement operations to output $f_k(x_1, x_2, x_3)$.
} of the $k$-party protocol $\Pi_k$, 
\begin{equation*}
\mathrm{QCC}_1(\Pi_k)= C_1 + C_2 + C_4,
\quad
\mathrm{QCC}_2(\Pi_k)= C_1 + C_3,
\quad
\mathrm{QCC}_3(\Pi_k)= C_2 + C_3 + C_4.
\end{equation*}
where $C_m$ denotes the number of qubits sent at the $m$-th round.
This value satisfies the equation
$2 \mathrm{QCC}(\Pi_k) = \sum_{i \leq k} \mathrm{QCC}_i(\Pi_k)$
where the factor of two comes from the fact that for each communication $C_m$, there are two players, one sending $C_m$ and one receiving $C_m$.
This equation implies that there is 
$i_0 \in [k]$ such that $\mathrm{QCC}_{i_0}(\Pi_{k}) \leq 2 \mathrm{QCC}(\Pi_k)/k$ (independent of the inputs).
(This is where the oblivious condition is used. If the protocol is not oblivious, the coordinate $i_0$ usually varies depending on the player's inputs.)
\par
For $i_0$, by the definition of the $k$-party-embedding property, there is a map $x_{-i_0}:\mathcal{X} \to \mathcal{X}^{k-1}$
such that $\tilde{f}_2(x_1, x_2) = f_k([x_{-i_0}(x_2), i_0, x_1])$ holds for any $x_1, x_2 \in \mathcal{X}$.
Using the protocol $\Pi_k$,  we then create a two-party protocol 
$\tilde{\Pi} \in \mathcal{P}_2(\tilde{f}_2, \varepsilon)$ with communication cost $\mathrm{QCC}_{i_0}(\Pi_k)$.
We name the two players in the protocol $\tilde{\Pi}$ Alice and Bob. Each is given $x_1, x_2$ respectively.
In the protocol $\tilde{\Pi}$, 
Alice plays the role of $P_{i_0}$ and Bob plays the other $k-1$ roles of $P_1, \ldots,P_{i_0 - 1}, P_{i_0 + 1},\ldots,P_k$.
Playing these roles, Alice and Bob simulate the original $\Pi_k$ with the input $[x_{-i_0}(x_2), i_0, x_1]$.
The communication cost of $\tilde{\Pi}$ is $\QCC_{i_0}(\Pi_k)$ because
the communication between Alice and Bob is made only when the player $P_{i_0}$ needs to communicate with others in the original protocol $\Pi_k$.
The other communications are internally computed by Bob.
(Fig~\ref{Fig_Pi_2} shows the two-party protocol $\tilde{\Pi}$ created from $\Pi_k$.)
When the simulation is finished, Alice and Bob can output the answer with error $\leq \varepsilon$ because
for the original protocol $\Pi_k$ for any $i \in [k]$, 
$\Pr(\Pi^i_\mathrm{out}([x_{-i_0}(x_2), i_0, x_1]) \neq f_k([x_{-i_0}(x_2), i_0, x_1]) \leq \varepsilon$ holds.
By the $k$-party-embedding property, we have $f_k([x_{-i_0}(x_2), i_0, x_1]) = \tilde{f}_2(x_1, x_2)$
which indicates $\tilde{\Pi} \in \mathcal{P}_2(\tilde{f}_2, \varepsilon)$ with communication cost $\mathrm{QCC}_{i_0}(\Pi_k) \leq \frac{2\mathrm{QCC}(\Pi_k)}{k}$.
\end{proof}

We also show a proposition which considers the bounded round setting.
\begin{proposition}\label{Prop_reduction_round}
Let $f_k$ and 
$\tilde{f}_2$ be the same as in Proposition~\ref{Prop_reduction}.
For any protocol $\Pi_k \in \mathcal{P}_k^M(f_k, \varepsilon)$,
there is a protocol $\tilde{\Pi} \in \mathcal{P}_2^M(\tilde{f}_2, \varepsilon)$ such that
$\mathrm{QCC}(\tilde{\Pi}) \leq \frac{\mathrm{QCC}(\Pi_k)}{k}$
holds.
\end{proposition}
\begin{proof}
In a similar manner as in Proposition~\ref{Prop_reduction}, 
we see that there is $i_0 \in [k]$ such that $\QCC_{i_0}(\Pi_k) \leq \frac{2}{k}\QCC(\Pi_k)$ holds.
(Note that in this case, we do not restrict the number of players communicating at each round.)
We create the desired two-party protocol $\tilde{\Pi}$ by Alice simulating $P_{i_0}$ and Bob simulating all the other players, except for $P_{i_0}$.
In the protocol $\tilde{\Pi}$, Alice and Bob need to communicate only if the player $P_{i_0}$ need to communicate with other players in the original protocol $\Pi_k$.
Therefore, the communication cost of the protocol satisfies
\begin{equation*}
\QCC(\tilde{\Pi}) = \sum_{m \in [M]} \sum_{j \in [k]\setminus \{i_0\}} C_{P_{i_0} \to P_j}(m) + C_{P_j \to P_{i_0}}(m) \leq \frac{2\QCC(\Pi_k)}{k}.
\end{equation*}
We also observe the protocol $\tilde{\Pi}$ is $M$-round protocol, completing the proof.
\end{proof}

Using Proposition~\ref{Prop_reduction}, we next show the following theorem.
\begin{theorem}[formal version]\label{Thm_reduction}
Let $f_{n, k}: \{0, 1\}^{n\cdot k} \to \{0, 1\}$
 be a $k$-party-embedding function of $\tilde{f}_n$.
Then
\begin{equation*}
\forall n, k, \quad \mathrm{QCC}(f_{n, k}, \varepsilon) \geq \frac{k}{2} \cdot \mathrm{QCC}(\tilde{f}_n, \varepsilon).
\end{equation*}
\end{theorem}
\begin{proof}
Let $\Pi_{n,k}$ be an optimal protocol for $f_{n, k}$, i.e., 
$\mathrm{QCC}(\Pi_{n,k}) =$ $\QCC(f_{n, k}, \varepsilon)$ $=\min_{\Pi \in \mathcal{P}_k(f_{n, k}, \varepsilon)}\QCC(\Pi)$.
By Proposition~\ref{Prop_reduction}, 
there is a two-party protocol $\tilde{\Pi} \in \mathcal{P}_2(\tilde{f}_n, \varepsilon)$ satisfying $\mathrm{QCC}(\tilde{\Pi}) \leq \frac{2 \mathrm{QCC}(\Pi_{n, k})}{k}$.
This yields
\begin{equation*}
\mathrm{QCC}(\tilde{f}_n, \varepsilon)
\leq \frac{2 \mathrm{QCC}(\Pi_{n, k})}{k}
= \frac{2 \mathrm{QCC}(f_{n, k}, \varepsilon)}{k}
\end{equation*}
which means
$\forall n, k, \quad \frac{k}{2}\QCC(\tilde{f}_n, \varepsilon) \leq \QCC(f_{n, k}, \varepsilon)$.
\end{proof}

\par
We can also prove a similar proposition in the bounded round scenario using Proposition~\ref{Prop_reduction_round}:
\begin{theorem}\label{Thm_reduction_round}
Let $f_{n, k}$
 and $\tilde{f}_n$ be the same as Theorem~\ref{Thm_reduction}.
 Then
for any n, k, $\mathrm{QCC}^M(f_{n, k}, \varepsilon) \geq \frac{k}{2} \cdot \mathrm{QCC}^{M}(\tilde{f}_n, \varepsilon)$
holds.
\end{theorem}
\begin{proof}
Note that in Proposition~\ref{Prop_reduction_round}, the new protocol for $\tilde{f}_n$ preserves 
the round of the original protocol $\Pi_k \in \mathcal{P}_k^M(f_{n, k}, \varepsilon)$.
Therefore in a similar manner as Theorem~\ref{Thm_reduction}, we get
\begin{equation*}
\forall n, k, \quad \frac{k}{2}\QCC^M(\tilde{f}_n, \varepsilon) \leq \QCC^M(f_{n, k}, \varepsilon).
\end{equation*}
\end{proof}

\section{Applications}\label{sec_application}
Here we investigate the lower bounds of some important functions such as Symmetric functions, Set-disjointness and Equality.
\par
We first apply Theorem~\ref{Thm_reduction} to symmetric functions.
Recall that
 any $k$-party symmetric function $f$ can be represented as $f(x_1, \ldots, x_k) = D_f(|x_1 \cap \cdots \cap x_k|)$
(each player is given $x_i (1 \leq i \leq k)$ as input)
using some function $D_f: [n]\cup\{0\} \to \{0, 1\}$.
\begin{proposition}\label{Pro_sym_lower}
$\mathrm{QCC}(f_{n, k}, 1/3) \in \Omega\left(k\{\sqrt{n l_0(D_{f_{n, k}})} + l_1(D_{f_{n, k}})\}\right)$ holds
for any $k$-party $n$-bit symmetric function $f_{n, k}$.
\end{proposition}
\begin{proof}
For $i \in [k]$, define 
$x_{-i}(x) := (x, 1^n, \ldots, 1^n) \in \{0, 1\}^{n \cdot (k -1)}$.
Then we have
that for any $i \in [k]$ and any $x_1, x_2 \in \{0, 1\}^n$, 
${f}_{n,2}(x_1, x_2)={f}_{n,k}([x_{-i}(x_2), i, x_1]).$
This implies $f_{n, k}$ is a $k$-party-embedding function of $f_{n, 2}$.
Therefore, Theorem~\ref{Thm_reduction} yields
$\mathrm{QCC}(f_{n, k}, 1/3)
\in \Omega(k \cdot \mathrm{QCC}(f_{n, 2}, 1/3)).$
Applying the well known lower bound $\Omega(\sqrt{n l_0(D_{f_{n, 2}})} + l_1(D_{f_{n, 2}}))$ of the two-party function $f_{n, 2}$~\cite{Raz03},
we obtain
\[\mathrm{QCC}(f_{n, k}, 1/3)
\in \Omega\left(k\{\sqrt{n l_0(D_{f_{n, k}})} + l_1(D_{f_{n, k}})\}\right).\]
\end{proof}
This lower bound is so strong that we get the optimal $\Omega(n \cdot k)$ bound for Inner-Product function (as $l_0(D_{f_{n, k}}) = l_1(D_{f_{n, k}}) = \Theta(n)$ holds) 
and $\Omega(k \sqrt{n})$ lower bound for Set-disjointness function (as $l_0(D_{f_{n, k}}) = 1$ and $l_1(D_{f_{n, k}}) = 0$ holds), 
which turns out to be optimal in our setting as described in Section~\ref{sec_match_bound}.

Next, we examine the lower bound of Equality function.
\begin{proposition}\label{Pro_Eq_lower}
$\mathrm{QCC}(\mathrm{Equality}_{n, k}, 1/3) \in \Omega(k)$.
\end{proposition}
\begin{proof}
For $i \in [k]$, define $x_{-i}: \{0, 1\}^n \to \{0, 1\}^{n \cdot (k - 1)}$ as
$x_{-i}(x) = (x, x, \ldots, x)$
(i.e., making $k-1$ copies of $x$).
Then we have
that for any $i \in [k]$, any $x_1, x_2 \in \{0, 1\}^n$, 
$\mathrm{Equality}_{n,2}(x_1, x_2)
=
\mathrm{Equality}_{n,k}([x_{-i}(x_2), i, x_1]).$
Therefore by Theorem~\ref{Thm_reduction},
the trivial lower bound $\Omega(1)$ of two-party $n$-bit Equality function
yields
$\mathrm{QCC}(\mathrm{Equality}_{n, k}, 1/3) \in \Omega(k).$
\end{proof}

We also prove a lower bound in bounded round scenario using Theorem~\ref{Thm_reduction_round}. 
\begin{proposition}\label{Pro_DISJ_m_lower}
$\mathrm{QCC}^M(\mathrm{DISJ}_{n, k}, 1/3) \in \Omega \left(n \cdot k / (M \log^8 M)\right)$.
\end{proposition}
\begin{proof}
Since the two-party $M$-round Set-disjointness requires $\Omega\left(n/(M \log^8 M)\right)$ communication~\cite{BGK+18},
we obtain
$\mathrm{QCC}^M(\mathrm{DISJ}_{n, k}, 1/3) \in \Omega \left(n \cdot k / (M \log^8 M)\right)$.
This is nearly tight as shown in Section~\ref{sec_match_bound}.
\end{proof}

\section{Matching upper bounds}\label{sec_match_bound}
In this section, we show the upper bound $O(k \sqrt{n})$ for $\mathrm{DISJ}_{n, k}$,
the upper bound $O(k\log n (\sqrt{n l_0(D_f)} + l_1(D_f)))$ for symmetric functions and the upper bound $O(k)$ for $\mathrm{Equality}_{n, k}$
 by creating efficient protocols for each function.
Without being noted explicitly, all of our protocols satisfy the oblivious routing condition.
These are (sometimes nearly) matching upper bounds since we have the same lower bounds in Section~\ref{sec_application}.

\subsection{Optimal protocol for \texorpdfstring{$\mathrm{DISJ}_{n, k}$}{DISJ}}
Here, we adopt the arguments from \cite[Section~7]{AA03}, which gives a two-party protocol for $\mathrm{DISJ}$ with $O(\sqrt{n})$-communication cost,
and present the protocol with $O(k \cdot \sqrt{n})$ cost in coordinator model.
\par
Let us first briefly describe the two-party protocol given in \cite{AA03}.
In the two-party protocol, 
inputs are represented as $(x_{ijk})_{(i, j, k) \in [n^{1/3}]^3} \in \{0, 1\}^n$ to Alice and $(y_{ijk})_{(i, j, k) \in [n^{1/3}]^3}\in \{0, 1\}^n$ 
to Bob.\footnote{If $n^{1/3}$ is not an integer, inputs are embedded to a larger cube of size $\lceil n^{1/3}\rceil^3$. 
In this case, for any coordinate $i \in \lceil n^{1/3}\rceil^3 \setminus [n]$, the $i$-th inputs $x_i$ and $y_i$ are set to $0$.}
They cooperate and communicate with each other 
to perform the following five operations (and their inverse operations) onto their registers:
\newline
Denoting the register for Alice (for Bob) as $|\psi \rangle_A$ ($|\psi \rangle_B$) and Alice holding an additional one qubit register $|z\rangle$,
\begin{itemize}
\item $O: |(i, j, k), z\rangle_A |(i, j, k)\rangle_B \mapsto |(i, j, k), z \oplus (x_{ijk}\wedge y_{ijk} )\rangle_A |(i, j, k)\rangle_B$
\item $W: |(i, j, k), z\rangle_A |(i, j, k)\rangle_B \mapsto (-1)^{z}|(i, j, k), z \rangle_A |(i, j, k)\rangle_B$
\item $S_V$: For a subset $V \subset [n^{1/3}]^3$,

\begin{equation*}
S_V: |(i, j, k), z\rangle_A |(i, j, k)\rangle_B \mapsto 
\left\{
\begin{array}{ll}
(-1)^{\delta_{0z}}|(i, j, k), z \rangle_A |(i, j, k)\rangle_B\\
\hspace{40mm}\text{if $(i, j, k) \in V$}\\
|(i, j, k), z\rangle_A |(i, j, k)\rangle_B\quad \hspace{3mm}\text{otherwise}
\end{array}.
\right.
\end{equation*}

\item For $d =1, 2, 3,$ 
\begin{equation*}
Z^d_\mathrm{plus}: |(i, j, k), z\rangle_A |(i, j, k)\rangle_B \mapsto
\begin{cases}
 |(i+1, j, k), z \rangle_A |(i+1, j, k)\rangle_B & \text{if $d=1$},\\
 |(i, j+1, k), z \rangle_A |(i, j+1, k)\rangle_B & \text{if $d = 2$},\\
 |(i, j, k+1), z \rangle_A |(i, j, k+1)\rangle_B & \text{if $d = 3$}.
\end{cases}
\end{equation*}

\item $Z^d_{\alpha, \beta}~(d = 1, 2, 3;~\alpha, \beta \in \mathbb{C} \text{~s.t.~} |\alpha|^2 + |\beta|^2 = 1)$\mbox{}\\
For specific subsets $V_1, V_2$ and $V_3$ (defined in the original paper~\cite{AA03}),
\begin{align*}
Z^1_{\alpha, \beta}: |(i, j, k), z\rangle_A |(i, j, k)\rangle_B \mapsto
\begin{cases}
 (\alpha |i\rangle^{\otimes 2}_{AB} + \beta|i+1\rangle^{\otimes 2}_{AB})|z \rangle_A |j, k\rangle^{\otimes 2}_{AB} \\
 \hspace{36mm}\text{if $(i, j, k) \in V_1$},\\
 \\
|(i, j, k), z\rangle_A |(i, j, k)\rangle_B \hspace{2mm}\text{otherwise}.
\end{cases}
\end{align*}
\begin{align*}
Z^2_{\alpha, \beta}: |(i, j, k), z\rangle_A |(i, j, k)\rangle_B \mapsto
\begin{cases}
 (\alpha |j\rangle^{\otimes 2}_{AB} + \beta|j+1\rangle^{\otimes 2}_{AB})|z \rangle_A |i, k\rangle^{\otimes 2}_{AB} \\
 \hspace{36mm}\text{if $(i, j, k) \in V_2$},\\
 \\
|(i, j, k), z\rangle_A |(i, j, k)\rangle_B \hspace{2mm}\text{otherwise}.
\end{cases}
\end{align*}
\begin{align*}
Z^3_{\alpha, \beta}: |(i, j, k), z\rangle_A |(i, j, k)\rangle_B \mapsto
\begin{cases}
 (\alpha |k\rangle^{\otimes 2}_{AB} + \beta|k+1\rangle^{\otimes 2}_{AB})|z \rangle_A |i, j\rangle^{\otimes 2}_{AB} \\
 \hspace{36mm}\text{if $(i, j, k) \in V_3$},\\
 \\
|(i, j, k), z\rangle_A |(i, j, k)\rangle_B \hspace{2mm}\text{otherwise}.
\end{cases}
\end{align*}

\end{itemize}
As shown in~\cite{AA03}, each operation is achieved by at most two qubits of communication:
$O$ and $Z^d_{\alpha, \beta}$ requires two qubits of communication and other operations $W, S_V$ and $Z^d_\mathrm{plus}$
are achieved without any communication.
In the two-party protocol, Alice and Bob use these operations $O(\sqrt{n})$ times to compute $\operatorname{Set-Disjointness}$.
Therefore in total $2 O(\sqrt{n}) = O(\sqrt{n})$ communication is sufficient in two-party case.
\par
In the following theorem, we explain how to extend these operations appropriately for the quantum multiparty communication model.
\begin{theorem}\label{Thm_DISJ_upper}
$\QCC(\operatorname{DISJ}_{n, k}, 1/3) \in O(k \sqrt{n})$.
\end{theorem}
\begin{proof}
Without loss of generality, we assume that the communication model is the coordinator model.
In our extension, the coordinator plays the role of Alice and $k$-players play the role of Bob.
For example, the query operation $O$ is extended to 
\begin{equation*}
O_k: |(i, j, l), z\rangle_\mathrm{Co} |(i, j, l)\rangle^{\otimes k}_{P_1 \cdots P_k}
\mapsto |(i, j, l), z \oplus (x^1_{ijl}\wedge \cdots \wedge x^k_{ijl} )\rangle_\mathrm{Co} |(i, j, l)\rangle^{\otimes k}_{P_1 \cdots P_k}.
\end{equation*}
Note that in this case each player $P_{i'}$ who is given an input $(x_{ijk}^{i'})$ holds the register $|i, j, l\rangle$.
We now explain how to extend each operation to that of coordinator model and how many qubits are needed to perform these operations.
\begin{itemize}
\item
$O_k: |(i, j, l), z\rangle_\mathrm{Co} |(i, j, l)\rangle^{\otimes k}_{P_1 \cdots P_k}
\mapsto |(i, j, l), z \oplus (\wedge_{i'\leq k}x^{i'}_{ijl})\rangle_\mathrm{Co} |(i, j, l)\rangle^{\otimes k}_{P_1 \cdots P_k}$
\newline
First, each player $P_{i'}$ performs $|(i, j, l)\rangle|0\rangle \mapsto |(i, j, l)\rangle |x_{ijl}\rangle$ using an auxiliary qubit $|0\rangle$.
Then they send the encoded qubits $|x_{ijl}^1\rangle \cdots|x_{ijl}^k\rangle$ to the coordinator who next performs
\[|(i, j, l), z\rangle|x^1_{ijl}, \ldots, x^k_{ijl}\rangle \mapsto |(i, j, l), z \oplus (\wedge_{i' \leq k}x^{i'}_{ijl})\rangle |x^1_{ijl}, \ldots, x^k_{ijl}\rangle\]
and return $|x^{i'}_{ijl}\rangle $ to each player $P_{i'}$.
Finally, each player clears the register: $|x^{i'}_{ijl}\rangle \mapsto |0\rangle$.
The total communication cost for this operation  is $2 k$ qubits.
\item $Z^1_{\alpha, \beta}: |(i, j, l), z\rangle_\mathrm{Co} |(i, j, l)\rangle^{\otimes k}_\mathrm{P_1 \cdots P_k} \mapsto
 \alpha |(i, j, l), z \rangle_{\mathrm{Co}} |(i, j, l)\rangle^{\otimes k}_\mathrm{P_1 \cdots P_k} +
 \beta |(i+1, j, l), z \rangle_\mathrm{Co} |(i+1, j, l)\rangle^{\otimes k}_\mathrm{P_1\cdots P_k}$ iff $(i, j, k)\in V_1$.
\newline
First, the coordinator creates $|0\rangle_C^{\otimes k} \mapsto \alpha |0\rangle_C^{\otimes k} + \beta |1\rangle_C^{\otimes k}$
from auxiliary qubits $|0\rangle_C^{\otimes k}$
and performs 
$|(i, j, l)\rangle_\Co (\alpha|0\rangle_C^{\otimes k} + \beta |1\rangle_C^{\otimes k}) \mapsto
\alpha |(i, j, l)\rangle_\Co |0\rangle_C^{\otimes k} + \beta |(i+1, j, l)\rangle_\Co |1\rangle_C^{\otimes k}.$
Next, the coordinator sends the auxiliary qubits to players, each player is given the single qubit.
On the received qubit $C_{i'}$ and the register $P_{i'}$, each player performs, 
for $a \in \{0, 1\}$, $|a\rangle_{C_{i'}} |(i, j, l)\rangle_{P_{i'}} \mapsto |a\rangle_{C_{i'}} |(i +a, j, l)\rangle_{P_{i'}}$.
They then return the auxiliary qubits to the coordinator who finally performs
$|(i+1, j, l)\rangle_\mathrm{Co} |1\rangle_C^{\otimes k} \mapsto |(i+1, j, l)\rangle_\mathrm{Co} |0\rangle_C^{\otimes k}$ iff $(i, j, k) \in V_1$.
The total communication for this operation is $2k$ qubits.
Other operations $Z^2_{\alpha, \beta}, Z^3_{\alpha, \beta}$ are achieved similarly.
\item The  operations $W, S, Z^d_\mathrm{plus}$ are done without any communication. 

\end{itemize}
\par
Suppose in the two-party protocol, Alice and Bob finally create the state $\sum_{(i, j, l)} \alpha_{ijl}|(i, j, l), z_{ijl}\rangle_A |(i, j, l)\rangle_B$
 applying the above operations $O(\sqrt{n})$ times. Then, with the same amount of steps, the coordinator and players can create the state
$\sum_{(i, j, l)} \alpha_{ijl} |(i, j, l), z_{ijl}\rangle_\mathrm{Co} |(i, j, l)\rangle^{\otimes k}_{P_1 \cdots P_k}$
whose amplitude $\{\alpha_{ijl}\}$ is the same as of the state in two party protocol.
Therefore, the coordinator can output the same answer as in the two-party protocol which implies that the success probability in the coordinator protocol
is the same as in the two-party protocol. After the coordinator obtain the answer, he/she finally send it to all players.

\par
Let us consider the communication cost needed to achieve this protocol.
In the coordinator model, there are $O(\sqrt{n})$ steps and each step needs at most $2k$ communication.
This shows $O(k \sqrt{n})$ upper bound of $\mathrm{DISJ}_{n, k}$ in the coordinator model.
\end{proof}

Using the protocol described in Theorem~\ref{Thm_DISJ_upper}, 
we can create $O(M)$-round protocol for $\mathrm{DISJ}_{n, k}$ with $O(n\cdot k/M)$ communication cost when $M \leq O(\sqrt{n})$.
The important fact here is that in the protocol with $O(k \sqrt{n})$ cost, the coordinator and players interact only for $O(\sqrt{n})$ rounds.
To create the desired protocol, let us now divide the input $x \in \{0, 1\}^n$ into $n/ M^2$ sub-inputs, each contains $M^2$ elements.
We next apply the above protocol \textit{in parallel} with the $n/M^2$ sub-inputs where each of sub-inputs uses $O(M)$ rounds and $O(k M)$ communication.
The new protocol still uses $O(M)$ rounds although the communication cost grows up to $\frac{n}{M^2} O(k M) = O(n \cdot k / M)$.
The success probability is still the same since the original protocol is a one-sided error protocol.
\par
Therefore, this protocol has $O(M)$ rounds and the communication cost $O(n \cdot k / M)$ which nearly matches the lower bound $\Omega\left(n \cdot k/(M \log^8 M)\right)$
described in Section~\ref{sec_application}. By converting this $M$-round coordinator protocol to the ordinary protocol, we obtain the following corollary:
\begin{corollary}\label{Cor_DISJ_m_upper}
$\mathrm{QCC}^{M}(\mathrm{DISJ}_{n, k}, 1/3) \in O(n \cdot k / M)$
when $M \leq O(\sqrt{n})$.
\end{corollary}

\subsection{Symmetric functions}\label{sec_upper_sym}
\begin{theorem}\label{Thm_sym_upper}
For any $k$-party $n$-bit symmetric function $f_{n, k}$,
\newline
\begin{centering}
$\mathrm{QCC}(f_{n, k}, 1/3) \in O\left(k\log n\{\sqrt{n l_0(D_{f_{n,k }})} + l_1(D_{f_{n, k}})\}\right)$.
\end{centering}
\end{theorem}
\begin{proof}
This proof is a generalization of \cite[Section~4]{Raz03} which investigates only the two-player setting.
Without loss of generality, we assume our model of communication to be the coordinator model.

\par
Let us first describe some important facts based on the arguments in~\cite{Raz03,She11}.
For any symmetric function $f_{n, k}$, the corresponding function $D_{f_{n, k}}$ is
constant on the interval $[l_0(D_{f_{n, k}}), n - l_1(D_{f_{n, k}})]$.
Without loss of generality, assume $D_{f_{n, k}}$ takes $0$ on the interval.
(If $D_{f_{n, k}}$ takes $1$ on the interval, we take the negation of $D_{f_{n, k}}$.)
Defining $D_0$ and $D_1:[n]\cup\{0\} \to \{0, 1\}$ as

\begin{equation*}
D_0(m) =
\begin{cases}
D_{f_{n, k}}(m) &\text{if $m \leq l_0$}\\
0  &\text{else}
\end{cases},~~
D_1(m) =
\begin{cases}
D_{f_{n,k}}(m) &\text{if $m > n - l_1$}\\
0  &\text{else}
\end{cases}
\end{equation*}
(abbreviating $ l_0 := l_0(D_{f_{n, k}})$ and $l_1:=l_1(D_{f_{n,k}})$),
$D_{f_{n, k}} = D_0 \vee D_1$ holds.
Therefore, by defining $f^0_{n, k}(x_1, \ldots, x_k) := D_0(|x_1 \cap \cdots \cap x_k|)$ and 
$f^1_{n, k}(x_1, \ldots, x_k) := D_1(|x_1 \cap \cdots \cap x_k|)$, we get $f_{n, k} = f^0_{n, k} \vee f^1_{n, k}$.
This means, computing $f^{0}_{n, k}$ and $f^{1}_{n, k}$ separately is sufficient to compute the entire function $f_{n, k}$.
As another important fact needed for our explanation, we note that the query complexity of $f^{0}_{n, k}$ equals to $O(\sqrt{n l_0(D_{f_{n, k}})})$
which is proven in \cite{Pat92}.
\par
Let us now explain a nearly optimal protocol for symmetric functions.
In this protocol, a coordinator computes $f_{n, k}$
by computing $f^{0}_{n, k}$and $f^{1}_{n, k}$ separately.
By the query complexity $O(\sqrt{n l_0(D_{f_{n, k}})})$ of the function $f^{0}_{n, k}$, the coordinator can compute $f^{0}_{n, k}$
by performing the query
$|i\rangle |y\rangle \mapsto |i\rangle |(x_1^i \wedge \cdots \wedge x_k^i) \oplus y\rangle~(1 \leq i \leq n)$
 for $O(\sqrt{n l_0(D_{f_{n, k}})})$ times.
We describe then how this query is implemented with $O(k \log n)$ communication.
For an $|i\rangle|y\rangle$, the procedure goes as follows.
\begin{enumerate}[(Step~1)]
\item Coordinator creates $k$ copies of $|i\rangle$: $ |i\rangle|y\rangle \mapsto |i\rangle^{\otimes k + 1} |y\rangle$
(using additional ancillary qubits to create $|i\rangle^{\otimes k}$)
and sends each of them to $k$ players.
\item Each player $j~(1 \leq j \leq k)$ of the $k$ players performs $|i\rangle |0\rangle \mapsto |i\rangle|x^i_j\rangle$
and sends the coordinator these qubits. Now the coordinator obtains $|i\rangle^{\otimes k + 1} |y\rangle |(x^i_1, \ldots, x^i_k)\rangle$
\item Coordinator performs $|y\rangle |(x^i_1, \ldots, x^i_k)\rangle \mapsto |(\wedge_{j \leq k}x^i_j) \oplus y\rangle |(x^i_1, \ldots, x^i_k)\rangle$
and return each $|i\rangle|x^i_j\rangle$ to player $j$.
\item Each player $j$ clears the register $|i\rangle|x^i_j\rangle \mapsto |i\rangle|0\rangle$ and returns $|i\rangle$.
Now the coordinator's register is  $|i\rangle^{\otimes k + 1} |(x^i_1\wedge \cdots \wedge  x^i_k) \oplus y\rangle$.
\end{enumerate}
This is how the query is implemented.
\par
Let us analyze how many qubits of communication is needed for this query. 
Step~1 requires $k \cdot \log n$ communication since $i \in [n]$ is represented by $\log n$ qubits.
Step~2 requires $k(\log n + 1)$ qubits by $\log n$ qubits for $i$ and one qubit for $x^i_j \in \{0, 1\}$.
Step~3 requires the same $k(\log n + 1)$ qubits and Step~4 requires $k \log n $ qubits.
Therefore, in total,  this query is implemented by $O(k \log n)$ qubits of communication and this protocol requires $O(k \log n \sqrt{n l_0(D_{f_{n, k}})})$ communication
to compute $f^{0}_{n, k}$.
\par
We next explain a protocol to compute $f^{1}_{n, k}$ which is simpler comparing to the protocol for $f^{0}_{n, k}$.
For the coordinator to compute $f^{1}_{n, k}$, each player $j$ tells the coordinator 
(1) if there are more than $\left(n - l_1(D_{f_{n, k}})\right)$ zeros and (2) where are zeros in the input $(x^i_j)_{i \leq n}$
when the first answer is YES (if the answer is NO, the player send an arbitrary bit string). 
This takes one qubit for the first question and $\log\left(\Sigma_{m = n - l_0(D_{f_{n, k}}) + 1}^n \binom{n}{m}\right) = O(l_1(D_{f_{n, k}}) \log n)$
qubits\footnote{
Here we use the fact that for any $n_0 \leq \frac{n}{2}$, $\log\left(\Sigma_{m = n - n_0 + 1}^n \binom{n}{m}\right) = O(n_0 \log n)$.
} for the second question. 
This needs $O(k l_1(D_{f_{n, k}}) \log n)$ communication in total.
With the information from players, the coordinator compute $f^{1}_{n, k}$ as follows.
First, if there is NO answered in the first question, the coordinator determines $f^{1}_{n, k} = 0$.
If every answer of the first question from players is YES, the coordinator calculates how many zeros are in $x_1 \cap \cdots \cap x_k \in \{0, 1\}^n$
which in turn gives the value of $|x_1 \cap \cdots \cap x_k|$. Therefore, the coordinator can compute 
the value of the function $f^{1}_{n, k}(x_1, \ldots, x_k) = D_1(|x_1 \cap \cdots \cap x_k|)$ even when there is no NO answer from  players.
\par
Combining these two protocols (one is for $f^{0}_{n, k}$ and the other is for $f^{1}_{n, k}$), the coordinator computes $f_{n, k}$
with $O(k \log n \sqrt{n l_0(D_{f_{n, k}})}) + O(k l_1(D_{f_{n, k}}) \log n) = O(k \log n \{\sqrt{n l_0(D_{f_{n, k}})} + l_1(D_{f_{n, k}})\})$ communication.
Finally, the coordinator sends the output to all players with the negligible $k$ bits of communication.
\end{proof}

\subsection{Optimal protocol for \texorpdfstring{$\mathrm{Equality}_{n, k}$}{Equality}}
Applying a public coin protocol with $O(1)$ communication cost for $\mathrm{Equality}_{n, 2}$ (see, e.g., \cite{KN96}) to the $k$-party case, we obtain the following proposition.
\begin{proposition}\label{Prop_Eq_upper}
$\mathrm{QCC}(\mathrm{Equality}_{n, k}, 1/3) \in O(k)$.
\end{proposition}

\section*{Acknowledgement}
FLG was partially supported by JSPS KAKENHI grants Nos.~JP16H01705, JP19H04066, JP20H00579, JP20H04139 and by MEXT Q-LEAP grants Nos.~JPMXS0118067394 and JPMXS0120319794. DS would like to take this opportunity to thank the “Nagoya University
Interdisciplinary Frontier Fellowship” supported by JST and Nagoya University. 
\clearpage
\bibliographystyle{unsrt}
\bibliography{Citations/LATIN22-20}
\end{document}